\documentclass[lettersize,journal]{IEEEtran}
\usepackage[T1]{fontenc} 
\usepackage{times}
\usepackage{amsmath}
\usepackage{amssymb}
\usepackage{tcolorbox}
\usepackage{amsthm}
\usepackage{color}
\usepackage{caption}   
\usepackage{algorithm}
\usepackage{algorithmic}
\usepackage{enumerate}
\usepackage{graphicx}
\usepackage{subfigure}
\usepackage{array}
\usepackage{booktabs}  
\usepackage{diagbox}   
\usepackage{multirow}  
\usepackage{cite}
\usepackage{upref}
\usepackage{bm}
\usepackage{arydshln}
\usepackage{mathtools}
\usepackage{microtype}
\usepackage{cuted}
\usepackage{wrapfig}
\usepackage{makecell}

\usepackage[figuresright]{rotating}
\usepackage{threeparttable}
\usepackage{booktabs}

\newtheorem{theorem}{Theorem}

\newtheorem{lemma}[theorem]{Lemma}

\long\def\symbolfootnote[#1]#2{\begingroup
\def\thefootnote{\fnsymbol{footnote}}\footnote[#1]{#2}\endgroup}
\renewcommand{\paragraph}[1]{{\bf #1}}
\title{Error Correction Decoding Algorithms of RS Codes Based on An Earlier Termination Algorithm to Find The Error Locator Polynomial}

\author{Zhengyi Jiang, Hao Shi, Zhongyi Huang, Linqi Song, Bo Bai, Gong Zhang, Hanxu Hou\\
}

\begin{document}
\let\emph\textit
\maketitle
\pagestyle{empty}  
\thispagestyle{empty} 
\symbolfootnote[0]{This paper was presented in part at the IEEE International Symposium on Information Theory (ISIT), 2024 \cite{2024ifdma}.
Z. Jiang, H. Shi and Z. Huang are with 
the Department of Mathematics Sciences, Tsinghua University, Beijing, China~(E-mail: jzy21@mails.tsinghua.edu.cn, shih22@mails.tsinghua.edu.cn, zhongyih@tsinghua.edu.cn). L. Song is with the Department of Computer Science, City University of Hong Kong (E-mail:linqi.song@cityu.edu.hk). B. Bai and G. Zhang are with the Theory Lab, Central Research Institute, 2012 Labs, Huawei Tech. Co. Ltd., Hong Kong SAR~(E-mail: baibo8@huawei.com, nicholas.zhang@huawei.com).
H. Hou is with the School of Electrical Engineering \& Intelligentization, Dongguan University of Technology~(E-mail: houhanxu@163.com). {\em (Corresponding author: Hanxu Hou.)}

This work was partially supported by the National Key R\&D Program of
China (No. 2020YFA0712300), the National Natural Science Foundation of China (No. 62071121, 62371411, 12025104),
Basic Research Enhancement Program of China under Grant 2021-JCJQ-JJ-0483.}
\begin{abstract}
Reed-Solomon (RS) codes are widely used to correct errors in storage systems. Finding the error locator polynomial is one of the key steps in the error correction procedure of RS codes.
Modular Approach (MA) is an effective algorithm for solving the Welch-Berlekamp (WB) key-equation problem to find the error locator polynomial that needs $2t$ steps, where $t$ is the error correction capability. In this paper, we first present a new MA algorithm that only requires $2e$ steps and then propose two fast decoding algorithms for RS codes based on our MA algorithm, where $e$ is the number of errors and $e\leq t$.
We propose Improved-Frequency Domain Modular Approach (I-FDMA) algorithm that needs $2e$ steps to solve the error locator polynomial and present our first decoding algorithm based on the I-FDMA algorithm. 
We show that, compared with the existing methods based on MA algorithms, our I-FDMA algorithm can effectively reduce the decoding complexity of RS codes when $e<t$.
Furthermore, we propose the $t_0$-Shortened I-FDMA ($t_0$-SI-FDMA) algorithm ($t_0$ is a predetermined even number less than $2t-1$) based on the new termination mechanism to solve the error number $e$ quickly. We propose our second decoding algorithm based on the SI-FDMA algorithm for RS codes and show that the multiplication complexity of our second decoding algorithm is lower than our first decoding algorithm (the I-FDMA decoding algorithm) when $2e<t_0+1$.

\end{abstract}


\section{Introduction}
\label{sec:1}
Reed-Solomon (RS) codes \cite{reed1960} are widely used in data storage systems because of their powerful error correction capability. An $(n,k)$ RS code encodes $k$ data symbols into $n$ coded symbols, and can correct any $e$ errors, where $e\leq t$ and $t=\lfloor \frac{n-k}{2} \rfloor$ is the error correction capability. Many decoding algorithms have been studied in \cite{1984ber,1995MA,bm1981,2003bm,2007bm,2012fbm,2016FFT,2020tang,2021f,2023eFDMA} to correct the errors for RS codes. 

Finding the error locator polynomial is the key part of the error correction process of RS codes.
The Berlekamp–Massey algorithm \cite{1984ber} and the Modular Approach (MA) \cite{1995MA} are two typical algorithms for solving error locator polynomial when decoding RS codes. The existing literature \cite{2007bm} has proved that the Berlekamp–Massey algorithm can stop the iteration and find the error locator polynomial in at least $t+e$ steps.

The key equation derived in the process of RS code error correction can be regarded as a special case of Welch-Berlekamp (WB) key equation, and the MA algorithm is effective for solving the WB key equation problem.
Based on MA method and LCH-FFT fast algorithm \cite{2016FFT}, the Frequency Domain Modular Approach (FDMA) algorithm is proposed in \cite{2022MA} to solve the error locator polynomial, which is suitable for decoding short RS codes and hardware implementation. 
The FDMA algorithm \cite{2022MA} needs to iterate $2t$ steps to get the error locator polynomial.
Recently, Chen {\em et al.} \cite{2023eFDMA} proposed an early termination mechanism for MA algorithm which can find the error locator polynomial by iterating only $t+e$ steps.
Similar to the FDMA algorithm, they \cite{2023eFDMA} proposed the enhanced FDMA (eFDMA) algorithm based on their termination mechanism and LCH-FFT algorithm.

The main contributions of this paper are summarized as follows.
\begin{itemize}
\item We first theoretically show that the MA algorithm can find the error locator polynomial by iterating only $2e$ steps when $e\leq t$ and propose a new 
termination mechanism to stop the algorithm at step $2e$.

\item We propose Improved-FDMA (I-FDMA) algorithm based on our new termination mechanism and present our first decoding algorithm based on I-FDMA algorithm. We show that our I-FDMA decoding algorithm has 5.64\% to 41.79\% multiplication reduction, compared with the existing eFDMA algorithm for evaluated parameters $(n,k)=(256,224)$ and $e\in\{1,2,\cdots,10\}$. 

\item We propose the $t_0$-Shortened I-FDMA (SI-FDMA) algorithm, which can be used to quickly find the number of errors $e$. 
Based on the $t_0$-SI-FDMA algorithm, we propose our second decoding algorithm, which has lower decoding complexity than our first decoding algorithm when  $2e<t_0+1$.
We show that our second decoding algorithm has 22.49\% to 46.41\% multiplication reduction for evaluated parameters $(n,k)=(256,224)$, $e\in\{1,2,\cdots,8\}$ and 26.71\% to 59.22\% multiplication reduction for evaluated parameters $(n,k)=(128,96)$, $e\in\{1,2,\cdots,8\}$, than our first decoding algorithm.
\end{itemize}

The rest of the paper is organized as follows. Section \ref{sec:2}
reviews the WB key equation problem and MA algorithm. Section \ref{sec:3} derives
the new termination mechanism for MA algorithm. Section \ref{sec:4} presents 
I-FDMA algorithm and our first decoding algorithm. 
Section~\ref{sec:five} presents the $t_0$-SI-FDMA algorithm and our second decoding algorithm.
Section~\ref{sec:com} evaluates our two decoding algorithms and the existing decoding algorithms based on MA methods. 
Section \ref{sec:6} concludes the paper.

\section{Preliminary}
\label{sec:2}
In this section, we first review the WB key equation problem \cite{2022MA}, then we introduce the classic MA algorithm and some related results.

Consider the binary finite fields $\mathbb{F}_{2^m}$ with $2^m$ elements. Let $\{v_i\}_{i=0}^{m-1}$ be the basis of $\mathbb{F}_{2^m}$ over $\mathbb{F}_{2}$, and $\mathbb{F}_{2^m}=\{\omega_i\}_{i=0}^{2^m-1}$, in which
\begin{align}
    \omega_i=i_0\cdot v_0+i_1\cdot v_1+\cdots+i_{m-1}\cdot v_{m-1}\nonumber
\end{align}
for each $i=i_0+i_1\cdot2+\cdots+i_{m-1}\cdot2^{m-1}\in\{0,1,\cdots,2^m-1\}$, where $(i_0,i_1,\ldots,i_{m-1})$ is the binary representation of $i$.

\subsection{The
WB Key Equation Problem and MA Algorithm}
\label{sec:2.1}
The Welch–Berlekamp key equation problem can be described as follows. Given that positive integer $\rho\geq 1$, find a polynomial pair $(\omega(x),n(x))$ with $\omega(x),n(x)\in \mathbb{F}_{2^m}[x]$ and $\deg(n(x))<\deg(\omega(x))$ such that the following equations hold,
\begin{align}
    n(x_i)=\omega(x_i)y_i, \forall i=1,2,\cdots,\rho,
    \label{eq.01}
\end{align}
where $\{(x_i,y_i)\}_{i=1}^{\rho}$ are given $\rho$ nonzero points over $\mathbb{F}_{2^m}$.

Consider an $(n,k)$ RS code over the field $\mathbb{F}_{2^m}$, we have that the error correction capability $t=\lfloor \frac{n-k}{2} \rfloor$ and $n\leq 2^m$.
In the error correction process of an $(n,k)$ RS code, find the error locator polynomial \cite{2022MA} is equivalent to find a polynomial pair $(\lambda(x),z(x))$ with $\deg(z(x))<\deg(\lambda(x))$ such that the following equation holds,
\begin{align}
    z(x)=s(x)\lambda(x)\mod \prod_{i=0}^{2t-1}(x-\omega_i),
    \label{eq.02}
\end{align}
where $s(x)\in \mathbb{F}_{2^m}[x]$ is the syndrome polynomial, 
$\lambda(x)\in \mathbb{F}_{2^m}[x]$ is the error locator polynomial and $z(x)\in \mathbb{F}_{2^m}[x]$ is
the error evaluation polynomial. The equation problem in Eq. \eqref{eq.02} is called the key equation problem. We can see that the key equation problem in Eq. \eqref{eq.02} is a special WB key equation problem in Eq. \eqref{eq.01} with $\rho=2t$ and $x_i=\omega_{i-1},y_i=s(\omega_{i-1})$ for $i=1,2,\cdots,2t$.

Suppose that there are $e$ errors in an $(n,k)$ RS code, where $e\leq t$. The element $\omega_i\in\mathbb{F}_{2^m}$ is the error position if and only if $\omega_i$ is the root of the error locator polynomial $\lambda(x)$, i.e., $\lambda(\omega_i)=0$. Let 
\begin{align}
    E:=\{i|\lambda(\omega_i)=0,i\in\{0,1,\cdots,n-1\}\}\nonumber
\end{align} be the index set of all $e$ error positions. Similar to the assumption in \cite{2023eFDMA}, suppose  that $E\subset\{i\}_{i=2t}^{n-1}$. In fact, we can always obtain an $(n,k)$ RS code as a shortened code of $(2^m,k+2^m-n)$ RS code by deleting the first $2^m-n$ codeword symbols.
If the number of shortened symbols $2^m-n$ is not less than $2t$, then the indices of both data symbols and parity symbols of $(n,k)$ RS code are $\geq2t$, and we can always have $E\subset\{i\}_{i=2t}^{n-1}$.

We can employ the MA method to solve the key equation problem in Eq. \eqref{eq.02}. It inputs $2t$ values $\{s(\omega_i)\}_{i=0}^{2t-1}$ and outputs polynomial pair $(\lambda(x),z(x))$ satisfying Eq. \eqref{eq.02}. We review the MA algorithm \cite{2022MA} in Algorithm \ref{alg.01}, which requires $2t$ iterations.
Please refer to the literature \cite{2022MA} for the idea behind the MA algorithm.

\begin{algorithm}[htb]
    \caption{The MA Algorithm \cite{2022MA}}
    \label{alg.01}
    \begin{algorithmic}[1]
        \REQUIRE $\{s(\omega_i)\}_{i=0}^{2t-1}$
        \ENSURE $(\lambda(x),z(x))$ satisfying $z(\omega_i)=s(\omega_i)\lambda(\omega_i)$ for $i=0,1,\cdots,2t-1$ and $\deg(\lambda(x))\leq t, \deg(z(x))<t$.
        \STATE Initialization:
        \STATE $\begin{pmatrix}
 w^{(0)}(x) & n^{(0)}(x)\\
 v^{(0)}(x) & m^{(0)}(x)
\end{pmatrix}=\begin{pmatrix}
 1 & 0\\
 0 & 1
\end{pmatrix}; \begin{pmatrix}
d^{(0)}_i \\
g^{(0)}_i
\end{pmatrix}=\begin{pmatrix}
-s(\omega_i) \\
1
\end{pmatrix},i=0,1,\cdots,2t-1; \begin{pmatrix}
R^{(0)}_0 \\
R^{(0)}_1
\end{pmatrix}=\begin{pmatrix}
0 \\
1
\end{pmatrix}$
    \FOR{$r=0,1,\cdots,2t-1$} 
    \IF{$d^{(r)}_r=0$ or ($(R^{(r)}_0>R^{(r)}_1)$ and $g^{(r)}_r\neq0$)}
        \STATE Let $\Psi_r=\begin{pmatrix}
 -g_r^{(r)} & d_r^{(r)}\\
 0 & x-\omega_{r}
\end{pmatrix}$
        \STATE and $\Psi_r(\omega_i)=\begin{pmatrix}
 -g_r^{(r)} & d_r^{(r)}\\
 0 & \omega_{i}-\omega_{r}
\end{pmatrix}$ for $i=r+1,r+2,\cdots,2t-1$
\STATE $R^{(r+1)}_0=R^{(r)}_0,R^{(r+1)}_1=R^{(r)}_1+2$
    \ELSE
    \STATE Let $\Psi_r=\begin{pmatrix}
 -g_r^{(r)} & d_r^{(r)}\\
 x-\omega_{r} & 0
\end{pmatrix}$
        \STATE and $\Psi_r(\omega_i)=\begin{pmatrix}
 -g_r^{(r)} & d_r^{(r)}\\
 \omega_{i}-\omega_{r} & 0
\end{pmatrix}$ for $i=r+1,r+2,\cdots,2t-1$
\STATE $R^{(r+1)}_0=R^{(r)}_1,R^{(r+1)}_1=R^{(r)}_0+2$
    \ENDIF
\FOR{$i=r+1,r+2,\cdots,2t-1$}
\STATE $\begin{pmatrix}
d^{(r+1)}_i \\
g^{(r+1)}_i
\end{pmatrix}=\Psi_r(\omega_i)\cdot\begin{pmatrix}
d^{(r)}_i \\
g^{(r)}_i
\end{pmatrix}$
\ENDFOR
\STATE $$\begin{pmatrix}
 w^{(r+1)}(x) & n^{(r+1)}(x)\\
 v^{(r+1)}(x) & m^{(r+1)}(x)
\end{pmatrix}=\Psi_r\cdot\begin{pmatrix}
 w^{(r)}(x) & n^{(r)}(x)\\
 v^{(r)}(x) & m^{(r)}(x)
\end{pmatrix}$$
    \ENDFOR
\IF{$R^{(2t)}_0<R^{(2t)}_1$}
\RETURN $(w^{(2t)}(x),n^{(2t)}(x))$ 
\ELSE
\RETURN $(v^{(2t)}(x),m^{(2t)}(x))$ 
\ENDIF
    \end{algorithmic}
\end{algorithm}

\subsection{Some Results of MA Algorithm}
\label{sec:2.2}
In the following, we review some existing results about MA algorithm in Algorithm \ref{alg.01}, which are useful for deriving our new termination mechanism in the next section.
The four variables $g_{r}^{(r)}$, $d_{r}^{(r)}$, $R_{0}^{(r)}$ and $R_{1}^{(r)}$ are used in Algorithm \ref{alg.01} which affect the iterations, where $r=0,1,\cdots,2t$.
The two variables $g_{r}^{(r)}$ and $d_{r}^{(r)}$ are in $\mathbb{F}_{2^m}$, while $R_{0}^{(r)}$ and $R_{1}^{(r)}$ are nonnegative integers, for $r=0,1,\cdots,2t$.
Please refer to the literature \cite{2022MA} for the detailed meaning of these four parameters.

The next lemma shows that $g_{r}^{(r)}$ and $d_{r}^{(r)}$ cannot be both zero.

\begin{lemma}\label{lem.01}
\cite[Lemma 4]{2022MA} Either $g_{r}^{(r)}$ or $d_{r}^{(r)}$ is nonzero for $r=0,1,\cdots,2t-1$.
\end{lemma}

Lemma \ref{lem.02} shows the relationship between $R_{0}^{(r)}$ and $R_{1}^{(r)}$ after each iteration.

\begin{lemma}\label{lem.02}
\cite[Lemma 5]{2022MA}$R_0^{(r)}+R_1^{(r)}=2r+1$ for $r=0,1,\cdots,2t$ and if $e \leq t$, $\min\{R_0^{(2t)},R_1^{(2t)}\}=2e$.
\end{lemma}

Lemma \ref{lem.03} reveals the key criterion about the termination mechanism of the MA algorithm.

\begin{lemma}\label{lem.03}
\cite[Lemma 1]{2023eFDMA} If $e\leq t$, then there exists some $r\leq t + e$ such that $R_1^{(r)}=2t+1$.
\end{lemma}

\begin{theorem}\label{th.01}
\cite[Theorem 1]{2023eFDMA} If $e\leq t$, assume that $E\subset\{i\}_{i=2t}^{n-1}$, then $(R_0^{(t+e)},R_1^{(t+e)})=(2e,2t+1)$, and the roots of $w^{(t+e)}
(x)$ are all $e$ error positions.
\end{theorem}
Together with Lemma \ref{lem.03} and Theorem \ref{th.01}, we can see that Algorithm \ref{alg.01} can be executed with $t+e$ steps.


\section{The New Termination Mechanism}
\label{sec:3}
In this section, we first show that we can execute the MA algorithm by only $2e$ steps to find the error locator polynomial. Note that $e$ is unknown in the error correction procedure, and we further derive a sufficient condition for the MA algorithm to be executed with $2e$ steps.

According to Lemma \ref{lem.02}, we can see that the sum of $R_0^{(i)}$ and $R_1^{(i)}$ is an odd number for $0\leq i\leq 2t$. Therefore, either $R_0^{(i)}$ or $R_1^{(i)}$ is an even number.
Let $\text{Even}(R_0^{(i)},R_1^{(i)})$ be the element in $\{R_0^{(i)},R_1^{(i)}\}$ which is an even number,
where $0\leq i\leq2t$. We can obtain the following lemma.

\begin{lemma}\label{lem.04}
For $0\leq i<j\leq 2t$, we have that $$\text{Even}(R_0^{(i)},R_1^{(i)})\leq \text{Even}(R_0^{(j)},R_1^{(j)}).$$ \end{lemma}
\begin{proof}
Without loss of generality, suppose that $R_0^{(i)}=\text{Even}(R_0^{(i)},R_1^{(i)})$, where $i\in\{0,1,\cdots,2t-1\}$.
If line 4 in Algorithm \ref{alg.01} is true, then line 7 will be performed, and we have $$R_0^{(i+1)}=R_0^{(i)}=\text{Even}(R_0^{(i+1)},R_1^{(i+1)}),$$ i.e., $\text{Even}(R_0^{(i+1)},R_1^{(i+1)})=\text{Even}(R_0^{(i)},R_1^{(i)})$.

On the other hand, if line 4 is false, then line 11 will be performed, and we have $$R_1^{(i+1)}=R_0^{(i)}+2=\text{Even}(R_0^{(i+1)},R_1^{(i+1)}),$$ i.e., 
$\text{Even}(R_0^{(i+1)},R_1^{(i+1)})>\text{Even}(R_0^{(i)},R_1^{(i)})$.
Therefore, the lemma is proved.
\end{proof}

The following lemma shows a sufficient condition for finding the $e$ error positions.
\begin{lemma}\label{lem.05}
If $R_0^{(i_0)}=2e$ and $R_1^{(i_0)}\geq 2e+1$, where $0\leq i_0\leq 2t$, then the roots of $w^{(i_0)}(x)$ are all $e$ error positions.
\end{lemma}
\begin{proof}
When $R_0^{(i_0)}=2e$ and $R_1^{(i_0)}\geq 2e+1$, we have $\text{Even}(R_0^{(i_0)},R_1^{(i_0)})=2e$.
According to Lemma \ref{lem.04} and Lemma \ref{lem.02}, we have
$2e=\text{Even}(R_0^{(i_0)},R_1^{(i_0)})\leq \text{Even}(R_0^{(i)},R_1^{(i)})\leq \text{Even}(R_0^{(2t)},R_1^{(2t)})=2e$, for any $i\geq i_0$. Therefore, for any $i\geq i_0$, we can obtain that $R_0^{(i)}=R_0^{(i_0)}=2e$. According to Lemma \ref{lem.02}, we have 
$R_0^{(i)}+R_1^{(i)}=2i+1$ for $i\geq i_0$, $R_1^{(i_0)}=2i_0+1-2e$, and further obtain that
\begin{align}
    R_1^{(i)}=2i+1-2e=2(i-i_0)+R_1^{(i_0)}\geq 2e+1>R_0^{(i)}.\label{eq.03}
\end{align}
Since $R_0^{(i)}=R_0^{(i_0)}=2e$ for any $i\geq i_0$, line 7 of Algorithm~\ref{alg.01} must be performed. Therefore, the condition in line 4 of Algorithm \ref{alg.01} must be true. 
Note that for $i\geq i_0$, we have $R^{(i)}_0<R^{(i)}_1$ by Eq. \eqref{eq.03}. Then we obtain that $d_i^{(i)}=0$; otherwise, the condition in line 4 of Algorithm \ref{alg.01} is false. We can further obtain that $g_i^{(i)}\neq0$ according to Lemma \ref{lem.01}. By line 16 of Algorithm \ref{alg.01}, for $i\geq i_0$, we have 
\begin{align}
&\begin{pmatrix}
 w^{(i+1)}(x) & n^{(i+1)}(x)\\
 v^{(i+1)}(x) & m^{(i+1)}(x)
\end{pmatrix}=\nonumber\\
&\begin{pmatrix}
 -g_i^{(i)} & 0\\
 0 & x-\omega_i
\end{pmatrix}\cdot\begin{pmatrix}
 w^{(i)}(x) & n^{(i)}(x)\\
 v^{(i)}(x) & m^{(i)}(x)
\end{pmatrix}.\nonumber
\end{align}
Thus $w^{(i+1)}(x)=-g_i^{(i)}w^{(i)}(x)$, and two polynomials $w^{(i+1)}(x)$ and $w^{(i)}(x)$ have the same roots.

To sum up, $w^{(i)}(x)$ and $w^{(t+e)}(x)$ have the same roots for any $i\geq i_0$. 
According to Theorem \ref{th.01}, we have that the roots of $w^{(t+e)}
(x)$ are all $e$ error positions. Therefore,  the roots of $w^{(i_0)}(x)$ are all $e$ error positions.
\end{proof}

According to Lemma \ref{lem.05}, we need to find a value of $i_0$ such that the condition in Lemma \ref{lem.05} holds. The next lemma shows that the condition in Lemma \ref{lem.05} holds when $i_0=2e$.

\begin{lemma}\label{lem.06}
When $i_0=2e$, we have $(R_0^{(2e)},R_1^{(2e)})=(2e,2e+1)$, i.e., the condition in Lemma \ref{lem.05} holds.
\end{lemma}
\begin{proof}
Let $r_0$ be the smallest integer such that $R_0^{(r_0)}=2e$. Obviously, $r_0\leq t+e$ since $R_0^{(t+e)}=2e$ according to Theorem~\ref{th.01}.
Next, we show that $R_1^{(r_0)}=2e+1$ by contradiction.

On the one hand, suppose that $R_1^{(r_0)}<2e+1$, then $$R_0^{(r_0)}+R_1^{(r_0)}=2r_0+1<2e+(2e+1),$$
where the above equality comes from Lemma \ref{lem.02}.
We obtain that $r_0<2e$ from the above inequality. According to Lemma~\ref{lem.04}, for any $2t-1\geq i\geq r_0$, we have 
\begin{align*}
2e=\text{Even}(R_0^{(r_0)},R_1^{(r_0)})&\leq \text{Even}(R_0^{(i)},R_1^{(i)})\\
&\leq \text{Even}(R_0^{(2t)},R_1^{(2t)})=2e.
\end{align*}
Therefore, $R_0^{(i)}=R_0^{(r_0)}=2e$ and $R_1^{(i)}=2i+1-2e$ for $2t\geq i\geq r_0$. Since $r_0<2e$, we have $(R_0^{(2e)},R_1^{(2e)})=(2e,2e+1)$.

On the other hand, suppose that $R_1^{(r_0)}>2e+1$. Note that $R_0^{(r_0)}=2e$ and $R_0^{(r_0-1)}<2e$. According to line 7 and line 11, we can obtain that $R_0^{(r_0-1)}=R_1^{(r_0)}-2$; otherwise, we have $R_0^{(r_0-1)}=R_0^{(r_0)}=2e$, which contradicts with $R_0^{(r_0-1)}<2e$.
Therefore, when $r=r_0-1$ in line 3, we will perform line 11 and the condition in line 4 is false.
Since $R_1^{(r_0)}$ is an odd number, we have $R_1^{(r_0)}\geq 2e+3$ and $$R_0^{(r_0-1)}=R_1^{(r_0)}-2\geq 2e+1>2e=R_1^{(r_0-1)},$$
then we have $d_{r_0-1}^{(r_0-1)}\neq 0$ and $g_{r_0-1}^{(r_0-1)}=0$.
By line 16 with $r=r_0-1$, we have 
\begin{align}
    &\begin{pmatrix}
 w^{(r_0)}(x) & n^{(r_0)}(x)\\
 v^{(r_0)}(x) & m^{(r_0)}(x)
\end{pmatrix}=\nonumber\\
&\begin{pmatrix}
 0 & d_{r_0-1}^{(r_0-1)}\\
 x-\omega_{r_0-1} & 0
\end{pmatrix}\cdot\begin{pmatrix}
 w^{(r_0-1)}(x) & n^{(r_0-1)}(x)\\
 v^{(r_0-1)}(x) & m^{(r_0-1)}(x)
\end{pmatrix},\nonumber
\end{align}
and $w^{(r_0)}(x)=d_{r_0-1}^{(r_0-1)}v^{(r_0-1)}(x)$. Therefore, two polynomials $w^{(r_0)}(x)$ and $v^{(r_0-1)}(x)$ have the same roots.
According to Lemma \ref{lem.05}, the roots of $w^{(r_0)}(x)$ are all $e$ error positions, thus the roots of $v^{(r_0-1)}(x)$ are also all error positions.

When $r=r_0-2$, if the condition in line 4 is true, then we have $v^{(r_0-1)}(x)=(x-\omega_{r_0-2})v^{(r_0-2)}(x)$ by line 16, which contradicts with the assumption that $v^{(r_0-1)}(x)$ does not contain the root whose index less than $2t$.
Otherwise, if the condition in line 4 is false, then we have $v^{(r_0-1)}(x)=(x-\omega_{r_0-2})w^{(r_0-2)}(x)$ by line 16, which also contradicts with the assumption that the root index of $v^{(r_0-1)}(x)$ is no less than $2t$.

Therefore, we have $R_1^{(r_0)}=2e+1$.
According to Lemma \ref{lem.02}, $R_0^{(r_0)}+R_1^{(r_0)}=2r_0+1=4e+1$, then we have $r_0=2e$. The lemma is proved.
\end{proof}

Combining Lemma \ref{lem.04} to Lemma \ref{lem.06}, we can know that $(R_0^{(2e)},R_1^{(2e)})=(2e,2e+1)$ and the roots of $w^{(2e)}(x)$ are all $e$ error positions when $e\leq t$ and $E\subset\{i\}_{i=2t}^{n-1}$. Therefore, we can obtain the error locator polynomial after $2e$ iterations in Algorithm \ref{alg.01}.
However, the value of $e$ is unknown. Next, we present a sufficient condition to find the error locator polynomial by only $2e$ iterations.

\begin{lemma}\label{lem.07}
Given an integer $i_1$ with $2t-1\geq i_1\geq 0$, if $d_{j}^{(i_1)}=0$ for all $2t-1\geq j\geq i_1$,
then $w^{(i_1)}(x)$ is the error locator polynomial.
\end{lemma}
\begin{proof}
Since $d_{j}^{(i_1)}=0$ for any $2t-1\geq j\geq i_1$, according to line 14 in Algorithm \ref{alg.01} with $i=j$, we have
\begin{align}
    d_{j}^{(i_1+1)}=-g_{i_1}^{(i_1)}d_{j}^{(i_1)}+d_{i_1}^{(i_1)}g_{j}^{(i_1)}=0,\nonumber
\end{align}
for any $2t-1\geq j\geq i_1+1$.
By analogy, we can know that $d_{j}^{(i)}=0$ for any $i_1\leq i\leq 2t-1$ and $i\leq j\leq 2t-1$ by repeatedly using line 14 in Algorithm \ref{alg.01}.

Therefore, for any $r\geq i_1$, we have $d_{r}^{(r)}=0$ and the condition in line 4 of Algorithm \ref{alg.01} is true. Then we have
\begin{align}
&\begin{pmatrix}
 w^{(r+1)}(x) & n^{(r+1)}(x)\\
 v^{(r+1)}(x) & m^{(r+1)}(x)
\end{pmatrix}=\nonumber\\
&\begin{pmatrix}
 -g_r^{(r)} & 0\\
 0 & x-\omega_r
\end{pmatrix}\cdot\begin{pmatrix}
 w^{(r)}(x) & n^{(r)}(x)\\
 v^{(r)}(x) & m^{(r)}(x)
\end{pmatrix},\nonumber
\end{align}
i.e., $w^{(r+1)}(x)=-g_r^{(r)}w^{(r)}(x)$. By Lemma \ref{lem.01}, we can see that $g_r^{(r)}\neq 0$ for $r\geq i_1$. Therefore, polynomials $w^{(r+1)}(x)$ and $w^{(r)}(x)$ have the same roots for any $r\geq i_1$. We can obtain that $w^{(i_1)}(x)$ and $w^{(2t)}(x)$ have the same roots. Since $w^{(2t)}(x)$ is the error locator polynomial, then $w^{(i_1)}(x)$ is also an error locator polynomial.

\end{proof}

Next, we show that the condition in Lemma \ref{lem.07} holds for $i_1=2e$.
\begin{lemma}\label{lem.08}
For any $2t-1\geq j\geq 2e$, we have $d_{j}^{(2e)}=0$.
\end{lemma}
\begin{proof}
According to Lemma \ref{lem.05} and Lemma \ref{lem.06}, we have that $R_{0}^{(2e)}=2e$, the roots of $w^{(2e)}(x)$ are all $e$ error positions, and $d_i^{(i)}=0$ for all $2e\leq i\leq 2t-1$. Next, we show that $d_{j}^{(2e)}=0$ for any $j\geq 2e$ by contradiction.

Suppose there exists $j_0$ with $2e\leq j_0\leq 2t-1$ such that $d_{j_0}^{(2e)}\neq 0$. According to line 14 in Algorithm \ref{alg.01}, we have $$d_{j_0}^{(2e+1)}=-g_{2e}^{(2e)}d_{j_0}^{(2e)}+d_{2e}^{(2e)}g_{j_0}^{(2e)}=-g_{2e}^{(2e)}d_{j_0}^{(2e)}.$$
Similarly, we can obtain 
\begin{align*}
d_{j_0}^{(2e+2)}&=-g_{2e+1}^{(2e+1)}d_{j_0}^{(2e+1)}+d_{2e+1}^{(2e+1)}g_{j_0}^{(2e+1)}\\&=-g_{2e+1}^{(2e+1)}d_{j_0}^{(2e+1)}\\&=(-g_{2e+1}^{(2e+1)})(-g_{2e}^{(2e)})d_{j_0}^{(2e)},
\end{align*}
and further, obtain 
\begin{align}
    d_{j_0}^{(j_0)}=\prod_{j=2e}^{j_0-1}(-g_{j}^{(j)})\cdot d_{j_0}^{(2e)}.\label{eq.04}
\end{align}

By Lemma \ref{lem.01}, we have $g_{j}^{(j)}\neq 0$ for all $2e\leq j\leq j_0-1$. Since $d_{j_0}^{(2e)}\neq 0$, i.e., the right side of Eq. \eqref{eq.04} is non-zero, which contradicts with $d_{j_0}^{(j_0)}=0$. Therefore, the lemma is finished.
\end{proof}

According to Lemma \ref{lem.07} and Lemma \ref{lem.08}, we only need to check whether $d_{j}^{(r)}=0$ for all $r\leq j\leq 2t-1$ in iteration $r$. If $d_{j}^{(r)}=0$ for all $r\leq j\leq 2t-1$ in iteration $r$, then we have $r=2e$ and we can stop Algorithm \ref{alg.01} to output the error locator polynomial.

\section{Our First Decoding Algorithm Based on I-FDMA Algorithm}
\label{sec:4}
In this section, we present our first decoding algorithm for LCH-FFT-based RS codes \cite{2016FFT}.
We first review the LCH-FFT related algorithms \cite{2016FFT,2022MA}, then propose the I-FDMA algorithm based on the new termination mechanism derived in Section \ref{sec:3}, and finally give our first decoding algorithm based on the LCH-FFT related algorithms and the I-FDMA algorithm.

\subsection{LCH-FFT Related Algorithms}\label{sec:4.1}
The encoding of LCH-FFT-based RS codes \cite{2016FFT} is based on LCH-FFT algorithm under LCH-basis $\bar{X}=\{\bar{X}_0(x),\bar{X}_1(x),\cdots,\bar{X}_{2^m-1}(x)\}$ in linear space $\mathbb{F}_{2^m}[x]/(x^{2^m}-x)$, where 
$$\bar{X}_{i}(x):=\frac{\prod_{j=0}^{m-1}(s_j(x))^{i_j}}{p_i},\ p_i=\prod_{j=0}^{m-1}(s_j(v_j))^{i_j}$$
for any $i=\sum_{j=0}^{m-1}i_j\cdot2^j\in\{0,1,\cdots,2^m-1\}$, and for any $j\in\{0,1,\cdots,m\}$, the subspace polynomial $s_j(x):=\prod_{i=0}^{2^j-1}(x-\omega_j)$.

The LCH-FFT algorithm is used to calculate the following estimation problem: given a polynomial $f(x)=\sum_{i=0}^{2^k-1}f_i\bar{X}_i(x)$ in $\mathbb{F}_{2^m}[x]/(x^{2^m}-x)$ with $\deg(f(x))<2^k$, $k\leq m$, and any $\beta\in\mathbb{F}_{2^m}$, calculate the values of $f(x)$ at $2^k$ points $\{\omega_{i}+\beta\}_{i=0}^{2^k-1}$, i.e., $\{f(\omega_i+\beta)\}_{i=0}^{2^k-1}$.
Let $\bm{f}_{\bar{X}}:=(f_0,f_1,\cdots,f_{2^k-1})$ represent the coefficient vector of polynomial $f(x)$ under basis $\bar{X}$.
We give the LCH-FFT algorithm in Algorithm \ref{FFT}. 
 
The inverse process of LCH-FFT algorithm, namely LCH-IFFT algorithm, corresponds to the interpolation problem: given $2^k$ values $\{d_i=f(\omega_i+\beta)\}_{i=0}^{2^k-1}$ in $\mathbb{F}_{2^m}$, where $\beta\in\mathbb{F}_{2^m}$ and $k\leq m$, calculate the coefficient vector $\bm{f}_{\bar{X}}$ of polynomial $f(x)=\sum_{i=0}^{2^k-1}f_i\bar{X}_i(x)$ based on basis $\bar{X}$.
We give the LCH-IFFT algorithm in Algorithm \ref{IFFT}. 

Tang {\em et al.} proposed the LCH-Extended IFFT algorithm \cite{2022MA} to solve the $2^k+1$-point interpolation problem, which is based on LCH-IFFT Algorithm \ref{IFFT} and can be applied to the decoding process of RS codes. We give the LCH-Extended IFFT algorithm in Algorithm \ref{extendedIFFT}. Let $h=2^k$, the complexity of each of Algorithms \ref{FFT}, \ref{IFFT} and  \ref{extendedIFFT} is $O(h\lg(h))$ \cite{2016FFT,2022MA}.

\begin{algorithm}[htb]
	\caption{$FFT_{\bar{X}}(\bm{f}_{\bar{X}},k,\beta)$\cite{2016FFT}}
	\label{FFT}
	\begin{algorithmic}[1]
	\REQUIRE $\bm{f}_{\bar{X}}=(f_0,f_1,\cdots,f_{2^k-1})$, $0\leq k\leq m,\beta\in \mathbb{F}_{2^m}$
	\ENSURE $2^k$ values: $(f(\omega_{0}+\beta),f(\omega_{1}+\beta),\cdots,f(\omega_{2^k-1}+\beta))$
		\IF {$k=0$}
		\RETURN $f_0$
		\ENDIF
		\FOR{$i=0,1,\cdots,2^{k-1}-1$} 
		\STATE $g_i^{(0)}=f_i+f_{i+2^{k-1}}\cdot\frac{s_{k-1}(\beta)}{s_{k-1}(v_{k-1})}$
		\STATE $g_i^{(1)}=g_i^{(0)}+f_{i+2^{k-1}}$
		\ENDFOR
		\STATE $V_0=FFT_{\bar{X}}(g^{(0)},k-1,\beta)$, where $g^{(0)}=(g_0^{(0)},g_1^{(0)},\cdots,g_{2^{k-1}-1}^{(0)})$
		\STATE $V_1=FFT_{\bar{X}}(g^{(1)},k-1,\beta)$, where $g^{(1)}=(g_0^{(1)},g_1^{(1)},\cdots,g_{2^{k-1}-1}^{(1)})$
		\RETURN $(V_0,V_1)$
	\end{algorithmic}
\end{algorithm}

\begin{algorithm}[htb]
	\caption{$IFFT_{\bar{X}}(\mathbf{D}_{2^k},k,\beta)$\cite{2016FFT}}
	\label{IFFT}
	\begin{algorithmic}[1]
		\REQUIRE $\mathbf{D}_{2^k}=(d_0,d_1,\cdots,d_{2^k-1})$, where $d_i=f(\omega_i+\beta)$ and $0\leq k\leq m,\beta\in \mathbb{F}_{2^m}$
		\ENSURE $\bm{f}_{\bar{X}}=(f_0,f_1,\cdots,,f_{2^k-1})$
		\IF {$k=0$}
		\RETURN $d_0$
		\ENDIF
		\STATE $(g^{(0)}_0,g^{(0)}_1,\cdots,g^{(0)}_{2^{k-1}-1})=IFFT_{\bar{X}}(V_0,k-1,\beta)$, where $V_0=(d_0,d_1,\cdots,d_{2^{k-1}-1})$
		\STATE $(g^{(1)}_0,g^{(1)}_1,\cdots,g^{(1)}_{2^{k-1}-1})=IFFT_{\bar{X}}(V_1,k-1,\beta)$, where $V_1=(d_{2^{k-1}},d_{2^{k-1}+1},\cdots,d_{2^k-1})$
		\FOR{$i=0,1,\cdots,2^{k-1}-1$} 
		\STATE $f_{i+2^{k-1}}=g_i^{(0)}+g_i^{(1)}$
		\STATE $f_i=g_i^{(0)}+\frac{s_{k-1}(\beta)}{s_{k-1}(v_{k-1})}\cdot f_{i+2^{k-1}}$
		\ENDFOR
		\RETURN $(f_0,f_1,\cdots,,f_{2^k-1})$
	\end{algorithmic}
\end{algorithm}

\begin{algorithm}[htb]
	\caption{$Extended\ IFFT_{\bar{X}}(k,\beta)$\cite{2022MA}}
	\label{extendedIFFT}
	\begin{algorithmic}[1]
		\REQUIRE $\{f(\omega_i+\beta)\}_{i=0}^{2^k}$, $0\leq k\leq m,\beta\in \mathbb{F}_{2^m}$
		\ENSURE $\mathbf{f}_{\bar{X}}=(f_0,f_1,\cdots,,f_{2^k})$
		\STATE Call Algorithm\ref{IFFT} with input $(f(\omega_0+\beta),f(\omega_1+\beta),\dots,f(\omega_{2^k-1}+\beta)),k,\beta$ to obtain $\hat{f}(x)$
		\STATE Calculate $\hat{f}(\omega_{2^k}+\beta)$
		\STATE $f(x)=(f(\omega_{2^k}+\beta)-\hat{f}(\omega_{2^k}+\beta))\cdot\bar{X}_{2^k}(x)+\hat{f}(x)-\frac{s_{k}(\beta)}{s_k(v_k)}\cdot(f(\omega_{2^k}+\beta)-\hat{f}(\omega_{2^k}+\beta))\cdot\bar{X}_{0}(x)$
		\RETURN $f(x)$
	\end{algorithmic}
\end{algorithm}

In the following, we consider the encoding of LCH-FFT-based $(n=2^m,k=2^m-2^\mu)$ RS codes over $\mathbb{F}_{2^m}$ whose error correction capability is $t=2^{\mu-1}$, where $\mu<m$. Let $T:=n-k=2^\mu$.
Any $1\times n$ codeword vector $\mathbf{F}$ can be represented as $\mathbf{F}=(\mathbf{F}_1,\mathbf{F}_2,\cdots,\mathbf{F}_{2^{m-\mu}})$, where $\mathbf{F}_i=(f(\omega_{(i-1)2^{\mu}}),f(\omega_{(i-1)2^{\mu}+1}),\cdots,f(\omega_{i2^{\mu}-1}))$ for any $i\in\{1,2,\cdots,2^{m-\mu}\}$.
The following lemma \cite[Lemma 10]{2016FFT} is the key to the encoding of systematic code.
\begin{lemma}\cite[Lemma 10]{2016FFT}
	\begin{align}
	    &IFFT_{\bar{X}}(\mathbf{F}_1,\mu,\omega_{0})+IFFT_{\bar{X}}(\mathbf{F}_2,\mu,\omega_{2^{\mu}})+\nonumber\\
     &\cdots+IFFT_{\bar{X}}(\mathbf{F}_{2^{m-\mu}},\mu,\omega_{2^m-2^{\mu}})=\mathbf{0}_{1\times 2^{\mu}},\label{eqsys}
	\end{align}
	in which the addition between two vectors means adding the elements corresponding to each position, and $\mathbf{0}_{a\times b}$ denotes $a\times b$ matrix with each element being 0.
	\label{systematicRS}
\end{lemma}
According to Lemma \ref{systematicRS}, we can perform systematic encoding for 
$(n=2^m,k=2^m-2^\mu)$ RS codes by setting $\mathbf{F}_2,\mathbf{F}_3,\cdots,\mathbf{F}_{2^{m-\mu}}$ as data symbols, and the parity symbol vector $\mathbf{F}_1$ can be calculated as follows
\begin{align*}
	\mathbf{F}_1&=FFT_{\bar{X}}(IFFT_{\bar{X}}(\mathbf{F}_2,\mu,\omega_{2^{\mu}})+\cdots\\
 &+IFFT_{\bar{X}}(\mathbf{F}_{2^{m-\mu}},\mu,\omega_{2^m-2^{\mu}}),\mu,\omega_0).
\end{align*}
Since the encoding process involves $(n/T-1)$ $T$-point LCH-FFT algorithms and one $T$-point LCH-IFFT algorithm, its complexity is $O(T\lg(T))+(n/T-1)O(T\lg(T))=O(n\lg(n-k)).$

\subsection{The I-FDMA Algorithm}
In the following, we present the I-FDMA algorithm for solving the error locator polynomial, which is based on our new termination mechanism of MA method in Section \ref{sec:3}. 

We show the I-FDMA algorithm in Algorithm \ref{alg.02}.
Recall that in Eq. \eqref{eq.02}, when $e\leq t$, we have $\deg(z(x))<\deg(\lambda(x))=e\leq t=2^{\mu-1}$.
Our I-FDMA algorithm only needs to update $t+1$ evaluations $\{w^{(r)}(\omega_i),v^{(r)}(\omega_i)\}_{i=0}^{t}$ of $w^{(r)}(x),v^{(r)}(x)$ for $r\in\{1,2,\cdots,2t\}$, and output the $t+1$ evaluations $\{\lambda(\omega_i)\}_{i=0}^{t}$ of the error locator polynomial. Note that our I-FDMA algorithm does not need to update the evaluations of two polynomials $n^{(r)}(x)$ and $v^{(r)}(x)$. Given an integer $r$ with $r\in\{0,1,\cdots,2t-1\}$, when $d_i^{(r)}=0$ for all $i=r,r+1,\cdots,2t-1$, then our I-FDMA algorithm can output $\{w^{(r)}(\omega_i)\}_{i=0}^{t}=\{\lambda(\omega_i)\}_{i=0}^{t}$. This is because the roots of $w^{(r)}(x)$ are the $e$ error positions when $r=2e$ according to Lemma \ref{lem.07} and Lemma \ref{lem.08}. 
Once $\{\lambda(\omega_i)\}_{i=0}^{t}$ are obtained, 
we can solve the error locator polynomial $\lambda(x)$ by employing Algorithm \ref{extendedIFFT}.

Note that in Algorithm \ref{alg.02}, line 14 requires three multiplications and two additions, and line 17 requires three multiplications and two additions. Therefore, the number of multiplications involved in I-FDMA algorithm is
\begin{align}\label{eqcomlexitymul}
    \sum_{r=0}^{2e-1}(3(2t-r-1)+3(t+1))=18et-6e^2+3e,
\end{align}
and the number of additions involved in I-FDMA algorithm is
\begin{align}\label{eqcomlexityadd}
    \sum_{r=0}^{2e-1}(2(2t-r-1)+2(t+1))=12et-4e^2+2e.
\end{align}

\begin{algorithm}[htb]
    \caption{The I-FDMA Algorithm}
    \label{alg.02}
    \begin{algorithmic}[1]
        \REQUIRE $\{s(\omega_i)\}_{i=0}^{2t-1}$
        \ENSURE $(\lambda^{(r+1)}(\omega_0),\lambda^{(r+1)}(\omega_1),\cdots,\lambda^{(r+1)}(\omega_t))$
        \STATE Initialization:
        \STATE $\begin{pmatrix}
d^{(0)}_i \\
g^{(0)}_i
\end{pmatrix}=\begin{pmatrix}
-s(\omega_i) \\
1
\end{pmatrix},i=0,1,\cdots,2t-1;$ 
\STATE $\begin{pmatrix}
W^{(0)}_i \\
V^{(0)}_i
\end{pmatrix}=\begin{pmatrix}
1 \\
0
\end{pmatrix},i=0,1,\cdots,t;$
\STATE $\begin{pmatrix}
R^{(0)}_0 \\
R^{(0)}_1
\end{pmatrix}=\begin{pmatrix}
0 \\
1
\end{pmatrix}.$
    \FOR{$r=0,1,\cdots,2t-1$} 
    \IF{($d^{(r)}_r=0$) or ($(R^{(r)}_0>R^{(r)}_1)$ and $g^{(r)}_r\neq0$)}
        \STATE Let $\Psi_r(\omega_i)=\begin{pmatrix}
 -g_r^{(r)} & d_r^{(r)}\\
 0 & \omega_{i}-\omega_{r}
\end{pmatrix}$ for $i=r+1,r+2,\cdots,2t-1$
\STATE $R^{(r+1)}_0=R^{(r)}_0,R^{(r+1)}_1=R^{(r)}_1+2$
    \ELSE
    \STATE Let $\Psi_r(\omega_i)=\begin{pmatrix}
 -g_r^{(r)} & d_r^{(r)}\\
 \omega_{i}-\omega_{r} & 0
\end{pmatrix}$ for $i=r+1,r+2,\cdots,2t-1$
\STATE $R^{(r+1)}_0=R^{(r)}_1,R^{(r+1)}_1=R^{(r)}_0+2$
    \ENDIF
\FOR{$i=r+1,r+2,\cdots,2t-1$}
\STATE $\begin{pmatrix}
d^{(r+1)}_i \\
g^{(r+1)}_i
\end{pmatrix}=\Psi_r(\omega_i)\cdot\begin{pmatrix}
d^{(r)}_i \\
g^{(r)}_i
\end{pmatrix}$
\ENDFOR
\FOR{$i=0,1,\cdots,t$}
\STATE $\begin{pmatrix}
W^{(r+1)}_i \\
V^{(r+1)}_i
\end{pmatrix}=\Psi_r(\omega_i)\cdot\begin{pmatrix}
W^{(r)}_i \\
V^{(r)}_i
\end{pmatrix}$
\ENDFOR
\IF{($d_{i}^{(r+1)}=0$ for $r+1\leq i\leq2t-1$) or ($r=2t-1$)}
\RETURN $(W^{(r+1)}_0,W^{(r+1)}_1,\cdots,W^{(r+1)}_{t})$
\ENDIF
    \ENDFOR
    \end{algorithmic}
\end{algorithm}

\subsection{Our First Decoding Algorithm}
In the following, we introduce our first decoding algorithm for LCH-FFT-based RS codes based on the I-FDMA algorithm. 

Assume that the received vector is 
\begin{align*}
    \mathbf{r}
    &=(r_0,r_1,\cdots,r_{2^m-1})\\
    &=\mathbf{F}+\mathbf{e}\\
    &=(f(\omega_0),f(\omega_1),\cdots,\omega_{2^m-1})+(e_0,e_1,\cdots,e_{2^m-1}),
\end{align*}
where $\mathbf{e}$ is the error pattern vector. The index set of error positions is $E=\{i|e_i\neq 0,i\in\{0,1,\cdots,2^m-1\}\}\subset[2t,2^m-1]$, and $|E|=e\leq t$.
Let $\mathbf{r}_i:=\{r_{i\cdot2^\mu},r_{i\cdot2^\mu+1},\cdots,r_{i\cdot2^\mu+2^\mu-1}\}$, for $i\in\{0,1,\cdots,2^{m-\mu}-1\}$.
Now, we introduce the decoding process.

First, we compute the syndrome polynomial $s(x)$ by Algorithm \ref{FFT} as follows \cite{2022MA},
\begin{align*}
	\sum_{i=0}^{2^\mu-1}IFFT_{\bar{X}}(\mathbf{r}_i,\mu,\omega_{i\cdot2^\mu})/p_{2^m-2^\mu}.
\end{align*}
Then we can calculate $2t$ syndromes $\{s(\omega_i)\}_{i=0}^{2t-1}$ (i.e., the input of Algorithm \ref{alg.02}) by $FFT_{\bar{X}}(\bm{s}_{\bar{X}},\mu,\omega_0)$.
Then we can get the $t+1$ values $\{\lambda(\omega_i)\}_{i=0}^{t}$ of the error locator polynomial $\lambda(x)$ by I-FDMA Algorithm \ref{alg.02}.
Next, according to Eq. \eqref{eq.02}, we can compute 
\begin{align}
z(\omega_i)=s(\omega_i)\lambda(\omega_i)=s(\omega_i)w(\omega_i),i=0,1,\cdots,t.\label{computez}
\end{align}
Then we can compute $z(x)$ and $\lambda(x)$ by Algorithm \ref{extendedIFFT}. 
After obtaining $\lambda(x)$, we can use the Chien search method to find the roots of $\lambda(x)$ by Algorithm \ref{FFT} as follows,
\begin{align}\label{achien}
FFT_{\bar{X}}(\bm{\lambda}_{\bar{X}},\mu,\omega_{i\cdot2^{\mu}}),\ i=0,1,\cdots,2^{m-\mu}-1.
\end{align}
Finally, we use Forney’s formula to compute the error values as follows,
\begin{align}
	f(\omega_\ell)-r(\omega_\ell)=\frac{z(\omega_\ell)}{s_\mu(\omega_\ell)\lambda'(\omega_\ell)}\ , \forall \ell\in E,\label{forney}
\end{align}
where $\lambda'(x)$ represents the formal derivative of $\lambda(x)$.
Please refer to \cite{2022MA} for the derivation of the above processes.

We summarize our first decoding algorithm in Algorithm~\ref{RSdecoding1}. The computational complexity of computing syndrome polynomial and syndromes, Chien search, formal derivative, and Forney’s formula is $O(n\lg(n-k))$ \cite{2022MA,2016FFT}.
Note that Algorithm \ref{RSdecoding1} can be generalized to any parameters $n$ and $k$ with asymptotic complexity $O(n\lg(n-k)+(n-k)\lg^2(n-k))$ \cite{Tang_2022}. Please refer to \cite{Tang_2022} for relevant details.

\begin{algorithm}[htb]
	\caption{Our First Decoding Algorithm}
	\label{RSdecoding1}
	\begin{algorithmic}[1]
		\REQUIRE Received vector $\mathbf{r}=\mathbf{F}+\mathbf{e}$
		\ENSURE The codeword $\mathbf{F}$
		\STATE Compute the syndrome polynomial $s(x)$ and syndromes $\{s(\omega_i)\}_{i=0}^{2^\mu-1}$
		\STATE Compute $\{\lambda(\omega_i)\}_{i=0}^{2^{\mu-1}}$ by I-FDMA Algorithm \ref{alg.02}, and get the error locator polynomial $\lambda(x)$ by Algorithm \ref{extendedIFFT}
		\STATE Compute $\{z(\omega_i)\}_{i=0}^{2^{\mu-1}}$, and get $z(x)$ by Algorithm \ref{extendedIFFT}
		\STATE Find all error positions by Chien search
		\STATE
		Compute the formal derivative of $\lambda(x)$, and then get the error pattern $\mathbf{e}$ by Forney's formula 
		\RETURN The codeword $\mathbf{F}=\mathbf{r}+\mathbf{e}$
	\end{algorithmic}
\end{algorithm}

\section{Our Second Decoding Algorithm}\label{sec:five}

In this section, we give our second decoding algorithm for RS codes. Our second decoding algorithm has lower decoding complexity than our first decoding algorithm when $e\ll t$. 
In the following, we still consider the LCH-FFT-based $(n=2^m,k=2^m-2^\mu)$ RS codes as in Section \ref{sec:4}.

The general idea of our second decoding algorithm is as follows: \textbf{Step (1):} we first determine the number of errors $e$ by the $t_0$-Shortened I-FDMA ($t_0$-SI-FDMA) algorithm (which is presented in the next subsection); \textbf{Step (2):} then we iterate the Early-Stopped Berlekamp–Massey (ESBM) algorithm \cite{2007bm} $2e$ steps to find the error locator polynomial; \textbf{Step (3):} finally, together with LCH-FFT related algorithms in Section \ref{sec:4.1}, we find the $e$ error locations and error values. 
The main difference between our second decoding algorithm and our first decoding algorithm in Algorithm \ref{RSdecoding1} is as follows. We employ I-FDMA algorithm to solve the error locator polynomial in our first decoding algorithm. While in our second decoding algorithm, we first propose the $t_0$-SI-FDMA algorithm to determine the number of errors $e$ and then employ ESBM algorithm to solve the error locator polynomial. We will show that our second decoding algorithm has lower multiplication complexity than our first decoding algorithm in Algorithm \ref{RSdecoding1} when $e\ll t$.

\subsection{The $t_0$-SI-FDMA Algorithm}
In the following, we propose the $t_0$-SI-FDMA algorithm that can quickly determine the number of errors $e$.

The idea behind the $t_0$-SI-FDMA algorithm is based on the following two observations. First, we can see that lines 16 to 18 in I-FDMA Algorithm \ref{alg.02} do not affect the values of $\{d_{r}^{(r)},g_r^{(r)},R_0^{(r)},R_1^{(r)}\}_{r=0}^{2t}$, which are the core parameters that affect the iterations of Algorithm \ref{alg.02}. Note that lines 16 to 18 in Algorithm \ref{alg.02} are just updating the $t+1$ values of the error locator polynomial.
Therefore, if we delete lines 16 to 18 from Algorithm \ref{alg.02}, and when the algorithm terminates, we output $R_0^{(r)}/2$ (note that $R_0^{(2e)}=2e$ according to Lemma \ref{lem.07}), then $R_0^{(r)}/2$ must be the number of errors $e$.
Second, for a given even number $t_0$ and $t_0<2t$, we will show (refer to Theorem \ref{tsi}) that Algorithm \ref{alg.02} does not need to input all $2t$ syndromes $\{s(\omega_{i})\}_{i=0}^{2t-1}$, but only $t_0+1$ syndromes $\{s(\omega_i)\}_{i=0}^{t_0}$, and we can still output the number of errors $e$ after $2e$ iterations when $2e<t_0+1$.

We give the $t_0$-SI-FDMA algorithm in Algorithm \ref{alg.03}.
The following theorem shows the key properties of Algorithm \ref{alg.03}.

\begin{algorithm}[htb]
    \caption{The $t_0$-SI-FDMA Algorithm}
    \label{alg.03}
    \begin{algorithmic}[1]
        \REQUIRE $\{s(\omega_i)\}_{i=0}^{t_0}$
        \ENSURE The number of errors $e$
        \STATE Initialization:
        \STATE $\begin{pmatrix}
d^{(0)}_i \\
g^{(0)}_i
\end{pmatrix}=\begin{pmatrix}
-s(\omega_i) \\
1
\end{pmatrix},i=0,1,\cdots,t_0;$ 
\STATE $\begin{pmatrix}
R^{(0)}_0 \\
R^{(0)}_1
\end{pmatrix}=\begin{pmatrix}
0 \\
1
\end{pmatrix}.$
    \FOR{$r=0,1,\cdots,t_0$} 
    \IF{($d^{(r)}_r=0$) or ($(R^{(r)}_0>R^{(r)}_1)$ and $g^{(r)}_r\neq0$)}
        \STATE Let $\Psi_r(\omega_i)=\begin{pmatrix}
 -g_r^{(r)} & d_r^{(r)}\\
 0 & \omega_{i}-\omega_{r}
\end{pmatrix}$ for $i=r+1,r+2,\cdots,t_0$
\STATE $R^{(r+1)}_0=R^{(r)}_0,R^{(r+1)}_1=R^{(r)}_1+2$
    \ELSE
    \STATE Let $\Psi_r(\omega_i)=\begin{pmatrix}
 -g_r^{(r)} & d_r^{(r)}\\
 \omega_{i}-\omega_{r} & 0
\end{pmatrix}$ for $i=r+1,r+2,\cdots,t_0$
\STATE $R^{(r+1)}_0=R^{(r)}_1,R^{(r+1)}_1=R^{(r)}_0+2$
    \ENDIF
\FOR{$i=r+1,r+2,\cdots,t_0$}
\STATE $\begin{pmatrix}
d^{(r+1)}_i \\
g^{(r+1)}_i
\end{pmatrix}=\Psi_r(\omega_i)\cdot\begin{pmatrix}
d^{(r)}_i \\
g^{(r)}_i
\end{pmatrix}$
\ENDFOR
\IF{$d_{i}^{(r+1)}=0$ for $r+1\leq i\leq t_0$}
\RETURN $R_0^{(r)}/2$
\ENDIF
    \ENDFOR
    \end{algorithmic}
\end{algorithm}

\begin{theorem}\label{tsi}
Let $t_0$ be an even number. When $2e<t_0+1$, Algorithm \ref{alg.03} can output the number of errors $e$; if the algorithm has no output, then there must be $2e>t_0+1$.
\end{theorem}
\begin{proof}
If $2e<t_0+1$, we can see that $d_{j}^{(2e)}=0$ for $j=2e,2e+1,\cdots,t_0$ by Lemma \ref{lem.08}, which means that line 15 of Algorithm \ref{alg.03} is true when $r=2e-1$, and Algorithm \ref{alg.03} must be terminated in no more than $2e$ iterations.
Note that the above analysis has shown that if Algorithm \ref{alg.03} does not output any value, it means that there must be $2e>t_0+1$.

Next, we only need to show that Algorithm \ref{alg.03} will be terminated strictly after $2e$ iterations.
We prove it by contradiction.
Assume that Algorithm \ref{alg.03} terminates after $u$ iterations, where $u<2e$. Then we have $d_{j}^{(u)}=0$ for $j=u,u+1,\cdots,t_0$.
According to line 13, we have 
\begin{align}
    d_{j}^{(u+1)}=-g_{u}^{(u)}d_{j}^{(u)}+d_{u}^{(u)}g_{j}^{(u)}=0,\nonumber
\end{align}
for $j=u+1,u+2,\cdots,t_0$.
By analogy, for any $a=u,u+1,\cdots,t_0$ and $b=a,a+1,\cdots,t_0$, we have $d_{b}^{(a)}=0$. By taking $a=2e-1<t_0$ and $b=a$, we have $d_b^{(a)}=d_{2e-1}^{(2e-1)}=0$.
Then, the judgment condition of line 5 in the $2e$-th iteration (i.e., $r=2e-1$) of Algorithm \ref{alg.03} is true, so that line 7 is executed, and thus there is $R_{0}^{(2e)}=R_{0}^{(2e-1)}$. However, according to Lemma \ref{lem.06}, the smallest integer $r_0$ such that $R_0^{(r_0)}=2e$ is $r_0=2e$, which contradicts with $R_{0}^{(2e)}=R_{0}^{(2e-1)}=2e$.
To sum up, Algorithm \ref{alg.03} will terminate strictly after 2e iterations and the theorem is proved.
\end{proof}

According to Theorem \ref{tsi}, given an even number $t_0$, 
the $t_0$-SI-FDMA algorithm is more suitable for the case where $e$ is relatively small (i.e., $2e<t_0+1$). If the $t_0$-SI-FDMA algorithm does not return any value, it means that $e$ is relatively large (i.e., $2e>t_0+1$), then we can employ the first decoding algorithm in Algorithm \ref{RSdecoding1}. Note that the values of $\{d_{j}^{(i)},g_{j}^{(i)}\}_{i=0,1,\cdots,t_0}^{j=i,i+1,\cdots,t_0}$ have been calculated by the $t_0$-SI-FDMA algorithm and they do not need to be calculated again when performing I-FDMA Algorithm in our first decoding algorithm. Therefore, the computational complexity of Algorithm \ref{RSdecoding1} will not be increased, when $2e>t_0+1$.

\subsection{Parity-Check Matrix and Syndromes}
In the following, unless otherwise specified, we assume that $2e<t_0+1$.
After \textbf{Step (1)}, we can get the number of errors $e$ by $t_0$-SI-FDMA algorithm.
We first analyze the parity-check matrix and the corresponding syndromes for LCH-FFT-based $(n=2^m,k=2^m-2^\mu)$ RS codes
to show that we can apply the ESBM algorithm \cite{2007bm} to find the error locator polynomial in \textbf{Step (2)}.

First, we introduce the key Theorem \cite[p.167]{104998} for deriving the parity-check matrix of LCH-FFT-based RS codes.

\begin{theorem}\cite[p.167]{104998}\label{paritycheck}
	Consider any $(n,k)$ RS code over $\mathbb{F}_{2^m}=\{\omega_i\}_{i=0}^{2^m-1}$ ($0<k<n\leq 2^m$) which is given by
 \begin{align*}
     \{f(\omega_0),f(\omega_1),\cdots,f(\omega_{n-1})|f(x)\in\mathbb{F}_{2^m}[x], \deg(f(x))<k\},
 \end{align*}
 then the dual code of the $(n,k)$ RS code can be given by
 \begin{align*}
     &\{\mu_0g(\omega_0),\mu_1g(\omega_1),\cdots,\mu_{n-1}g(\omega_{n-1})|\\
     &g(x)\in\mathbb{F}_{2^m}[x], \deg(g(x))<n-k\},
 \end{align*}
 where $\mu_i=\frac{1}{\prod_{i\neq j}(\omega_i-\omega_j)}$.
\end{theorem}

According to Theorem \ref{paritycheck}, when $n=2^m$, $\{\omega_i\}_{i=0}^{n-1}$ form a linear space over $\mathbb{F}_2$, then $\mu_0=\mu_1=\cdots=\mu_{n-1}$, and the parity-check matrix of LCH-FFT-based $(n=2^m,k=2^m-2^\mu)$ RS codes $\mathbf{H}$ is an $(n-k)\times n$ Vandermonde matrix, specially, 
\begin{align}\label{parityH}
     \mathbf{H}=\begin{bmatrix}
 1 & 1 & 1 & 1& 1\\
 \omega_0 & \omega_1 & \omega_2 & \cdots & \omega_{2^m-1}\\
 \omega_0^2 & \omega_1^2 & \omega_2^2 & \cdots & \omega_{2^m-1}^2\\
 \vdots & \vdots & \vdots & \ddots & \vdots\\
 \omega_0^{n-k-1} & \omega_1^{n-k-1} & \omega_2^{n-k-1} & \cdots & \omega_{2^m-1}^{n-k-1}
\end{bmatrix}.
 \end{align}

Based on Eq. \eqref{parityH}, we have
\begin{align}\label{syndrome}
    \mathbf{H}\cdot\mathbf{r}^T=\mathbf{H}\cdot\mathbf{F}^T+\mathbf{H}\cdot\mathbf{e}^T=\mathbf{H}\cdot\mathbf{e}^T.
\end{align}
Let $\mathbf{H}\cdot\mathbf{r}^T:=(S_0,S_1,\cdots,S_{2t-1})^T$, and $\{S_i\}_{i=0}^{2t-1}$ be the $2t$ syndromes. 
Suppose there are $e<\frac{t_0}{2}$ errors, and $e_{i_j}\neq 0$ for any $j\in[e]$, where $E=\{i_j\}_{j\in[e]}\in[2t,n-1]$.
Let $\beta_j:=\omega_{i_j}$ and $\delta_j=e_{i_j}$, for $j\in[e]$.
According to Eq. \eqref{syndrome}, we have 
\begin{align}\label{eqsss}
    S_i=\beta_1^i\delta_1+\beta_2^i\delta_2+\cdots+\beta_e^i\delta_e, \ \forall i\in\{0,1,\cdots,2t-1\}.
\end{align}
Let $$\sigma(x):=\prod_{j\in[e]}(1-\beta_jx)=1+\sigma_1x+\sigma_2x^2+\cdots+\sigma_ex^e$$
be the error locator polynomial of our second decoding method (which differs from that of Algorithm \ref{RSdecoding1}).
Once the unknown coefficients $\{\sigma_i\}_{i\in[e]}$ of $\sigma(x)$ are found, all error positions can be obtained by searching the roots of $\sigma(x)$.

According to the definition of $\sigma(x)$ and Eq. \eqref{eqsss}, the relationship between the syndromes and the coefficients of the error locator polynomial $\sigma(x)$ can be obtained as follows \cite[p.261]{Moon2005ErrorCC},
\begin{align}\label{eqsss2}
    \sum_{\ell=0}^{e}\sigma_{\ell}S_{j+e-\ell}=0,\ \forall j\in\{0,2t-1-e\},
\end{align}
where, we define $\sigma_0=1$.

Let $\mathbf{S}$ be the $t\times (t+1)$ Hankel matrix \cite{2007bm} formed by all $2t$ syndromes $\{S_{i}\}_{i=0}^{2t-1}$, i.e., 
\begin{align*}
    \mathbf{S}=\begin{pmatrix}
	S_0  & S_1 & \cdots & S_{t-1}&S_t\\
	S_1 & S_2 & \cdots & S_t&S_{t+1}\\
	\vdots & \vdots & \ddots& \vdots&\vdots\\
	S_{t-1} & S_t & \cdots & S_{2t-2}&S_{2t-1}\\
\end{pmatrix}.
\end{align*}
Next, we show that the $e\times e$ sub-matrix in the upper left corner of $\mathbf{S}$ is invertible.

\begin{lemma}\label{inv-S}
    The matrix $\mathbf{S}_{e\times e}$ is invertible, where $\mathbf{S}_{e\times e}$ represents the $e\times e$ sub-matrix in the upper left corner of $\mathbf{S}$, i.e.,
    $$\mathbf{S}_{e\times e}=\begin{pmatrix}
	S_0  & S_1 & \cdots & S_{e-1}\\
	S_1 & S_2 & \cdots & S_e\\
	\vdots & \vdots & \ddots& \vdots\\
	S_{e-1} & S_e & \cdots & S_{2e-2}\\
\end{pmatrix}.$$
\end{lemma}
\begin{proof}
    Let $\mathbf{W}$ be the $e\times e$ Vandermonde matrix as follows,
    $$\mathbf{W}:=\begin{pmatrix}
	1  & 1 & \cdots & 1\\
	\beta_1 & \beta_2 & \cdots & \beta_e\\
	\vdots & \vdots & \ddots& \vdots\\
	\beta_1^{e-1} & \beta_2^{e-1} & \cdots & \beta_e^{e-1}\\
\end{pmatrix}.$$
Then we can verify that 
$$\mathbf{W}\cdot\text{diag}(\delta_1,\delta_2,\cdots,\delta_e)\cdot\mathbf{W}^T=\mathbf{S}_{e\times e},$$
where $\text{diag}(\delta_1,\delta_2,\cdots,\delta_e)$ represents the diagonal matrix whose main diagonal elements are $\delta_1,\delta_2,\cdots,\delta_e$ respectively.

Therefore, we have 
$$\det(\mathbf{S}_{e\times e})=\det(\mathbf{W})^2\cdot\prod_{j=1}^{e}\delta_j.$$
Note that $\delta_j=e_{i_j}\neq 0$, and the Vandermonde matrix $\mathbf{W}$ is invertible since $\beta_i\neq \beta_j$ for $i\neq j\in[e]$.
We have $\det(\mathbf{S}_{e\times e})\neq 0$, i.e., $\mathbf{S}_{e\times e}$ is invertible. 
\end{proof}

According to Eq. \eqref{eqsss2}, we have
\begin{align}\label{onlyone}
    \mathbf{S}_{e\times e}\cdot\begin{bmatrix}
 -\sigma_e\\
 -\sigma_{e-1}\\
 \vdots\\
-\sigma_1
\end{bmatrix}=\begin{bmatrix}
 S_e\\
 S_{e+1}\\
 \vdots\\
 S_{2e-1}
\end{bmatrix}.
\end{align}
Recall that $\mathbf{S}_{e\times e}$ is invertible by Lemma \ref{inv-S}, thus Eq. \eqref{onlyone} has a unique solution, and we can see that we only need to calculate $2e$ syndromes $\{S_i\}_{i=0}^{2e-1}$ to solve the $e$ coefficients $\{\sigma_i\}_{i=1}^{e}$ of the error locator polynomial $\sigma(x)$.
 
\subsection{The S-ESBM Algorithm}
In the following, we review the Early-Stopped Berlekamp–Massey (ESBM) algorithm \cite{2007bm}, which can be used to efficiently solve Eq. \eqref{onlyone} to get the error locator polynomial $\sigma(x)$.

We summarize the problem that can be solved by ESBM algorithm in Theorem \ref{ESBM}. Please refer to \cite{2007bm} for more details.
\begin{theorem}\cite{2007bm}\label{ESBM}
    Given a positive integer $t$ and a $t\times (t+1)$ Hankel matrix $\mathbf{S}$, where
\begin{align*}
    \mathbf{S}=\begin{pmatrix}
	S_0  & S_1 & \cdots & S_{t-1}&S_t\\
	S_1 & S_2 & \cdots & S_t&S_{t+1}\\
	\vdots & \vdots & \ddots& \vdots&\vdots\\
	S_{t-1} & S_t & \cdots & S_{2t-2}&S_{2t-1}\\
\end{pmatrix}.
\end{align*}
For $i\in\{1,2,\cdots,t+1\}$, let $\mathbf{\xi}_i$ represent the $i$-th column of $\mathbf{S}$.
If there exists a positive integer $e\leq t$, and an $e\times 1$ unknown vector $(\Lambda_e,\Lambda_{e-1},\cdots,\Lambda_1)^T$, such that:
\begin{itemize}
    \item [(1)] The $e\times e$ sub-matrix in the upper left corner of $\mathbf{S}$ is invertible;
    \item [(2)] For any $1\leq j\leq t+1-e$, 
    $$\Lambda_e\mathbf{\xi}_j+\Lambda_{e-1}\mathbf{\xi}_{j+1}+\cdots+\Lambda_1\mathbf{\xi}_{e+j-1}+\mathbf{\xi}_{e+j}=\mathbf{0}_{t\times 1}.$$
\end{itemize}
Then, if $e$ is unknown, the ESBM algorithm needs to iterate $t+e$ steps to obtain $(\Lambda_e,\Lambda_{e-1},\cdots,\Lambda_1)^T$;
if $e$ is known, the ESBM algorithm needs to iterate $2e$ steps to obtain $(\Lambda_e,\Lambda_{e-1},\cdots,\Lambda_1)^T$.
\end{theorem}

According to Lemma \ref{inv-S}, the matrix $\mathbf{S}_{e\times e}$ is invertible. According to Eq. \eqref{eqsss2}, we can verify that for any $1\leq j\leq t+1-e$, 
$$\sigma_e\mathbf{\xi}_j+\sigma_{e-1}\mathbf{\xi}_{j+1}+\cdots+\sigma_1\mathbf{\xi}_{e+j-1}+\mathbf{\xi}_{e+j}=\mathbf{0}_{t\times 1}.$$
Therefore, the two conditions of Theorem \ref{ESBM} are satisfied, and we can use the ESBM algorithm to solve Eq. \eqref{onlyone} and obtain $(\sigma_e,\sigma_{e-1},\cdots,\sigma_1)$.
Note that the number of errors $e$ is known according to the $t_0$-SI-FDMA algorithm, so the ESBM algorithm only needs $2e$ syndromes $\{S_i\}_{i=0}^{2e-1}$ and $2e$ iterations to output the $e$ unknown coefficients of $\sigma(x)$. 
Recall that ESBM algorithm \cite{2007bm} needs $t+e$ steps to solve the error locator polynomial, where $e$ is unknown. We propose the Shortened-ESBM (S-ESBM) algorithm which is given in Algorithm \ref{S-ESBM}, that can solve the error locator polynomial with $2e$ steps, where $e$ is known. In Algorithm \ref{S-ESBM}, for any $1\times n$ vector, the notation $\mathbf{v}|_{\ell}$ ($\ell\leq n$) represents the right $\ell$-truncate of the vector $\mathbf{v}$, i.e., $\mathbf{v}|_{\ell}=(v_{n-\ell+1},\cdots,v_{n-1},v_{n})$.
The multiplication complexity of S-ESBM is $2e^2-1$ \cite{2007bm}.

\begin{algorithm}[htb]
    \caption{The S-ESBM Algorithm}
    \label{S-ESBM}
    \begin{algorithmic}[1]
    \REQUIRE $e$, $\{S_i\}_{i=0}^{2e-1}$
    \ENSURE The coefficient vector of error locator polynomial $\mathbf{\Lambda}$
    \FOR{$r = 1$ to $2e$}
\STATE Calculate $d_r$ by $d_r=\mathbf{\Lambda}\cdot(S_{r-L},\dots,S_r)$
    
\IF{$d_r = 0$}
\STATE$\tilde{\mathbf{\Lambda}} \gets (\tilde{\mathbf{\Lambda}},0)$
    
\ELSIF{$d_r \neq 0$ and $r\leq 2L$}
    \STATE $\mathbf{\Lambda} \gets \mathbf{\Lambda}-d_r\cdot D^{-1}\cdot(\tilde{\mathbf{\Lambda}},0)|_{L+1}$
    \STATE$\tilde{\mathbf{\Lambda}} \gets (\tilde{\mathbf{\Lambda}},0)$
    
   \ELSIF{$d_r \neq 0$ and $r > 2L$}
   \STATE $\mathbf{\Lambda}_{temp} \gets \mathbf{\Lambda}$
   \STATE $\mathbf{\Lambda} \gets (\mathbf{0}_{1\times (r-2L)},\mathbf{\Lambda}) - d_r\cdot D^{-1}\cdot(\tilde{\mathbf{\Lambda}},0)|_{r-L+1}$
   \STATE $\tilde{\mathbf{\Lambda}} \gets (\mathbf{0}_{1\times(r-2L)},\mathbf{\Lambda}_{temp})$
   \STATE $D \gets d_r$
   \STATE $L \gets r - L$
   
\ENDIF
\ENDFOR
\RETURN $\mathbf{\Lambda}$
    \end{algorithmic}
\end{algorithm}

{\em Remark:} Note that the operation in line 17 of I-FDMA Algorithm \ref{alg.02} requires three field multiplications. Therefore, compared with the $t_0$-SI-FDMA algorithm, solving the error locator polynomial by I-FDMA algorithm requires at least 
$$2e\cdot3(t+1)=6e(t+1)$$
more multiplications.
When $e\ll t$, we have $2e^2-1\ll 6e(t+1)$, which means that the new algorithm for solving the error locator polynomials (i.e., using the $t_0$-SI-FDMA algorithm and S-ESBM algorithm) will further effectively reduce the multiplication complexity for the case of $e\ll t$, which is also a major reason for the multiplication complexity reduction in our second decoding algorithm.

\subsection{Transformation of Coefficient Vector}

Note that the coefficient vector $(1,\sigma_1,\sigma_2,\cdots,\sigma_e)$ of $\sigma(x)$ is based on the standard basis $\{1,x,x^2,\cdots,x^{2^m-1}\}$ in $\mathbb{F}_{2^m}/(x^{2m}-x)$. In order to employ the LCH-FFT algorithm in Section \ref{sec:4.1} in our second decoding algorithm, we need to convert the coefficient vector into the form based on the LCH-basis $\bar{X}$.
According to the definition of $\bar{X}$, there exists $(e+1)\times (e+1)$ invertible matrix $\mathbf{T}_{e+1}$ such that
\begin{align}\sigma(x)&=(1,\sigma_1,\sigma_2,\cdots,\sigma_e)\cdot\begin{bmatrix}
 1\\
 x\\
 \vdots\\
 x^e
\end{bmatrix}\nonumber\\
    &=(1,\sigma_1,\sigma_2,\cdots,\sigma_e)\cdot\mathbf{T}_{e+1}\cdot\begin{bmatrix}
 \bar{X}_0(x)\\
 \bar{X}_1(x)\\
 \vdots\\
 \bar{X}_{e}(x)
\end{bmatrix}.\label{convert}
\end{align}
Note that $\mathbf{T}_{e+1}$ is a lower triangular matrix and can be pre-calculated, and the element in the first row and the first column of $\mathbf{T}_{e+1}$ is 1 since $\bar{X}_0(x)=1$, so the number of multiplications of the coefficient vector transformation of $\sigma(x)$ is less than or equal to $\frac{e^2+3e}{2}$.
To distinguish, in the following, for any polynomial $f(x)$, let $\bm{f}_{\bar{X}}$ be the coefficient vector of $f(x)$ under the basis $\bar{X}$ and $\bm{f}$ be the coefficient vector under the standard basis.

\subsection{Our Second Decoding Algorithm}
Let $s:=\left \lfloor \log(e) \right \rfloor +1$ and $R:=2^s$, then $e<R\leq 2e<t_0+1$.

We present decoding process for \textbf{Step (3)} of our second decoding algorithm as follows.
Once $\bm{\sigma}_{\bar{X}}$ is calculated by S-ESBM Algorithm \ref{S-ESBM} and Eq.~\eqref{convert}, we can find the roots $\{\beta_j^{-1}=\omega_{i_j}^{-1}\}_{j=1}^{e}$ of $\sigma(x)$ through the Chien search and FFT Algorithm \ref{FFT} as follows,
\begin{align}\label{chien}
    FFT_{\bar{X}}(\bm{\sigma}_{\bar{X}},s,\omega_{\ell\cdot 2^s}),\ \forall \ell=0,1,\cdots,2^{m-s}-1,
\end{align}
and the complexity is $O(n\lg(R))$.
Then we can calculate the error locator polynomial $\lambda(x)=\prod_{j=1}^{e}(x-\omega_{i_j})$ (here, we still call $\lambda(x)$ as the error locator polynomial as in Section \ref{sec:4}), in which the number of multiplications is $\frac{e^2-e}{2}$.
Then we can get $\bm{\lambda}_{\bar{X}}$ according to Eq. \eqref{convert}.
Next, we can obtain $\bm{z}_{\bar{X}}$ by following three steps:
first, we compute $\{\lambda(\omega_{i})\}_{i=0}^{R-1}$ by $FFT_{\bar{X}}(\bm{\lambda}_{\bar{X}},s,\omega_{0})$; second, we compute $\{z(\omega_{i})\}_{i=0}^{R-1}$ by Eq. \eqref{computez}; finally, we get $\bm{z}_{\bar{X}}$ by $IFFT_{\bar{X}}(\mathbf{D}_{R},s,\omega_0)$, where $\mathbf{D}_{R}:=(z(\omega_{0}),z(\omega_{1}),\cdots,z(\omega_{R-1}))$, and the 
complexity of these three steps is $O(n\lg(R))$.
Finally, we can obtain the error values $\{e_{i_j}\}_{j=1}^{e}$ by Forney's formula by Eq. \eqref{forney}, and the complexity is $O(n\lg(R))$.

We summarize our second decoding algorithm in Algorithm \ref{RSdecoding2}. 
Note that we can use the Reed-Muller (RM) transform \cite{2018RM,2020RM,2023RM} to calculate the $2e$ syndromes in line 3 in Algorithm \ref{RSdecoding2} to further reduce the number of multiplications. When $e=1$, the number of multiplications is $3\left \lfloor \log(n) \right \rfloor -1$ \cite{2018RM}, and when $e\geq 2$, there is no exact expression for the number of multiplications of RM transform. The asymptotic multiplications of RM transform for $e\geq 2$ is given in \cite{2023RM}, please refer to \cite{2023RM} for more details.

\begin{algorithm}[htb]
	\caption{Our Second Decoding Algorithm}
	\label{RSdecoding2}
	\begin{algorithmic}[1]
		\REQUIRE Received vector $\mathbf{r}=\mathbf{F}+\mathbf{e}$
		\ENSURE The codeword $\mathbf{F}$
		\STATE Compute the syndrome polynomial $s(x)$ and syndromes $\{s(\omega_i)\}_{i=0}^{2^\mu-1}$
		\STATE Compute the number of errors $e$ by $t_0$-SI-FDMA Algorithm \ref{alg.03}
        \STATE Compute $2e$ syndromes $\{S_i\}_{i=0}^{2e-1}$
		\STATE Compute the coefficient vector based on the standard basis of $\sigma(x)$ by S-ESBM Algorithm \ref{S-ESBM}, and obtain $\sigma_{\bar{X}}$
		\STATE Find all error locations by Chien search
		\STATE Compute $\{\lambda(\omega_{i})\}_{i=0}^{R-1}$ by $FFT_{\bar{X}}(\bm{\lambda}_{\bar{X}},s,\omega_{0})$, then we compute $\{z(\omega_{i})\}_{i=0}^{R-1}$ by Eq. \eqref{computez}, and get $\bm{z}_{\bar{X}}$ by IFFT Algorithm \ref{IFFT}
		\STATE Compute the formal derivative of $\lambda(x)$, and then get the error pattern $\mathbf{e}$ by Forney's formula 
		\RETURN The codeword $\mathbf{F}=\mathbf{r}+\mathbf{e}$
	\end{algorithmic}
\end{algorithm}

\section{Comparative Analysis}
\label{sec:com}
\subsection{Comparison of MA-based Algorithms}
\label{sec:5}
In the following, we focus on evaluating the decoding complexity for our I-FDMA algorithm and eFDMA algorithm \cite{2023eFDMA}. The difference between our I-FDMA algorithm and eFDMA algorithm is that our I-FDMA algorithm only needs $2e$ steps to find the $e$ error positions, while eFDMA algorithm requires $t+e$ steps. Therefore, the decoding complexity of our I-FDMA algorithm is strictly less than that of eFDMA algorithm when $e<t$.

First, we show the number of multiplications and the number of additions involved in finding the $e$ error positions by FDMA algorithm \cite{2022MA}, eFDMA algorithm \cite{2023eFDMA}, and I-FDMA algorithm in Table \ref{tab0} for $(n,k)$ RS code with error correction capability $t$, where $e\leq t$.

\begin{table}[htpb]
    \centering
    \caption{The number of multiplications and additions involved in finding the $e$ error positions by FDMA algorithm \cite{2022MA}, eFDMA algorithm \cite{2023eFDMA} and I-FDMA algorithm for $(n,k)$ RS code with error correction capability $t$, where $e\leq t$.} 
   
\renewcommand\arraystretch{2.0}
    \resizebox{0.40\textwidth}{!}{\begin{tabular}{|c|c|c|}
\hline  
	\makecell{MA-based\\ algorithms} & \makecell{Multiplication \\complexity} & \makecell{Addition\\ complexity}  \\
 \hline  
 FDMA\cite{2022MA}&$12t^2+3t$&$8t^2+2t$\\
 \hline
 eFDMA\cite{2023eFDMA}&\makecell{$\frac{1}{2}(7t^2+26et+$
 \\$3(e+t)-9r^2)$}&\makecell{$\frac{1}{2}(5t^2+16et-$\\$5e^2+3e+t)$}\\
 \hline
Our I-FDMA&$18et-6e^2+3e$&$12et-4e^2+2e$\\
 \hline
\end{tabular}}  
\label{tab0}
\end{table}

According to Table \ref{tab0}, when $e\ll t$, the complexity of I-FDMA algorithm is only $O(t)$ level, while the complexity of FDMA and eFDMA algorithms is both $O(t^2)$ level.

Because the multiplication in finite fields is more time-consuming than the addition, we show the number of multiplications involved in finding the $e$ error positions by our I-FDMA algorithm and eFDMA algorithm for $(n=256,k=224)$ RS code in Table \ref{tab2}.
The results in Table \ref{tab2} demonstrate that our I-FDMA algorithm can reduce the number of multiplications by 9.94\%-74.67\%, compared with eFDMA algorithm.

\begin{table}[htpb]
    \centering
    \caption{The number of multiplications involved in finding the error positions by I-FDMA algorithm and eFDMA algorithm for  $(256,224)$ RS code. In the table, the improvement is defined as the ratio of the multiplication reduction to the number of multiplications involved in finding the error positions by eFDMA algorithm.} 
\renewcommand\arraystretch{1.5}
    \resizebox{0.40\textwidth}{!}{
\begin{tabular}{|c|c|c|c|}
	\hline  
	$e$ (t=16)&eFDMA&I-FDMA&Improvement\\
	\hline  
	1 & 1125 & 285 & 74.67\% \\
	\hline  
	2 & 1321 & 558 & 57.76\% \\
	\hline  
	3 & 1508 & 819 & 45.69\% \\
	\hline  
	4 & 1686 & 1068 & 36.65\% \\
	\hline  
	5 & 1855 & 1305 & 29.65\% \\
	\hline  
	6 & 2015 & 1530 & 24.07\% \\
	\hline  
	7 & 2166 & 1743 & 19.53\% \\
	\hline  
	8 & 2308 & 1944 & 15.77\% \\
	\hline  
	9 & 2441 & 2133 & 12.62\% \\
	\hline  
	10 & 2565 & 2310 & 9.94\% \\
	\hline  
\end{tabular}}  
\label{tab2}
\end{table}

Moreover, we show the total number of multiplications involved in our first decoding algorithm in Algorithm \ref{RSdecoding1} of $(256,224)$ RS code and eFDMA algorithm in Tabel \ref{tab3}. The results show that our first decoding algorithm can reduce the number of multiplications by 5.64\%-41.79\%, compared with that based on eFDMA algorithm.

\begin{table}[htpb]
    \centering
    \caption{The number of multiplications involved in our first decoding algorithm in Algorithm \ref{RSdecoding1} and eFDMA algorithm for  $(256,224)$ RS code. In the table, the improvement is defined as the ratio of the multiplication reduction to the number of multiplications involved in the decoding algorithm based on eFDMA algorithm.} 
\renewcommand\arraystretch{1.5}
    \resizebox{0.40\textwidth}{!}{
\begin{tabular}{|c|c|c|c|}
	\hline  
	$e$ (t=16)&eFDMA&Our Algorithm \ref{RSdecoding1}&Improvement\\
	\hline  
	1 & 2010 & 1170 & 41.79\% \\
	\hline  
	2 & 2352 & 1589 & 32.44\% \\
	\hline  
	3 & 2549 & 1860 & 27.03\% \\
	\hline  
	4 & 2935 & 2317 & 21.05\% \\
	\hline  
	5 & 3130 & 2580 & 17.57\% \\
	\hline  
	6 & 3316 & 2831 & 18.12\% \\
	\hline  
	7 & 3493 & 3070 & 14.62\% \\
	\hline  
	8 & 4125 & 3761 & 12.10\% \\
	\hline  
	9 & 4324 & 4016 & 8.82\% \\
	\hline  
	10 & 4514 & 4259 & 5.64\% \\
	\hline  
\end{tabular}}  
\label{tab3}
\end{table}
\subsection{Algorithm \ref{RSdecoding2} VS Algorithm \ref{RSdecoding1}}
Next, we compare the multiplication complexity for the proposed two decoding algorithms, namely our first decoding algorithm in Algorithm \ref{RSdecoding1} (based on I-FDMA algorithm) and our second decoding algorithm in Algorithm \ref{RSdecoding2} (based on $t_0$-SI-FDMA algorithm and S-ESBM algorithm).
We take $t_0=t$ in the following comparisons.

In Table \ref{tab30}, we give the number of multiplications of decoding Algorithm \ref{RSdecoding1} and decoding Algorithm \ref{RSdecoding2} for $(n=256,k=224)$ RS code error correction procedure, and we take $e\in\{1,2,\cdots,8\}$.
In Table \ref{tab40}, we give the number of multiplications of decoding Algorithm \ref{RSdecoding1} and decoding Algorithm \ref{RSdecoding2} for $(n=128,k=96)$ RS code error correction procedure, and we take $e\in\{1,2,\cdots,8\}$.

\begin{table}[htpb]
    \centering
    \caption{The number of multiplications of decoding Algorithm \ref{RSdecoding1} and decoding Algorithm \ref{RSdecoding2} for $(n=256,k=224)$ RS code error correction decoding process. In the table, the improvement is defined as the ratio of the multiplication reduction to the number of multiplications by Algorithm \ref{RSdecoding1}.}
    \renewcommand\arraystretch{1.5}
    \resizebox{0.40\textwidth}{!}
   {\begin{tabular}{|c|c|c|c|}
	\hline  
	$e$ (t=16)&\makecell{Decoding\\ Algorithm \ref{RSdecoding1}}&\makecell{Decoding\\ Algorithm \ref{RSdecoding2}}&Improvement\\
	\hline  
	1 & 1170 & 769 & 34.27\% \\
	\hline  
	2 & 1589 & 1057 & 33.48\% \\
	\hline  
	3 & 1860 & 1193 & 35.86\% \\
	\hline  
	4 & 2317 & 1591 & 31.33\% \\
	\hline  
	5 & 2580 & 1707 & 33.83\% \\
	\hline  
	6 & 2831 & 1928 & 31.89\% \\
	\hline  
	7 & 3070 & 2117 & 31.04\% \\
	\hline  
	8 & 3761 & 2931 & 22.06\% \\
	\hline  
\end{tabular}}
\label{tab30}
\end{table}

\begin{table}[htpb]
    \centering
    \caption{The number of multiplications of decoding Algorithm \ref{RSdecoding1} and decoding Algorithm \ref{RSdecoding2} for $(n=128,k=96)$ RS code error correction decoding process. In the table, the improvement is defined as the ratio of the multiplication reduction to the number of multiplications by Algorithm \ref{RSdecoding1}.}
    \renewcommand\arraystretch{1.5}
    \resizebox{0.40\textwidth}{!}
   {\begin{tabular}{|c|c|c|c|}
	\hline  
	$e$ (t=16)&\makecell{Decoding\\ Algorithm \ref{RSdecoding1}}&\makecell{Decoding\\ Algorithm \ref{RSdecoding2}}&Improvement\\
	\hline  
	1 & 786 & 449 & 42.87\% \\
	\hline  
	2 & 1141 & 673 & 41.01\% \\
	\hline  
	3 & 1412 & 809 & 42.70\% \\
	\hline  
	4 & 1805 & 1143 & 36.67\% \\
	\hline  
	5 & 2068 & 1259 & 39.11\% \\
	\hline  
	6 & 2319 & 1480 & 36.17\% \\
	\hline  
	7 & 2558 & 1669 & 34.75\% \\
	\hline  
    8 & 3185 & 2419 & 24.05\% \\
	\hline
\end{tabular}}
\label{tab40}
\end{table}

According to Table \ref{tab30} and Table \ref{tab40}, when $e\leq t/2$, the results show that Algorithm \ref{RSdecoding2} can reduce the number of multiplications by 22.06\%-35.86\% for parameters $n=256,k=224$ and by 24.05\%-42.87\% for parameters $n=128,k=96$, compared with Algorithm \ref{RSdecoding1}.

To sum up, our second decoding algorithm in Algorithm \ref{RSdecoding2} is more suitable for the case of $2e<t_0+1$.
We summarize the essential reasons as follows: (1) I-FDMA algorithm needs to multiply matrix and vector to iterate the value of the error locator polynomial, which will cause much more multiplication operations when $e\ll t$, while the $t_0$-SI-FDMA Algorithm \ref{tsi} and the S-ESBM Algorithm \ref{S-ESBM} have lower multiplication complexity when $e\ll t$; (2) when $e\ll t$, the RM transform can effectively reduce the multiplication complexity for syndrome calculation (i.e., line 3 of Algorithm \ref{RSdecoding2}), while with the increase of $e$, the number of multiplications required by RM transform will also increase rapidly \cite{2023RM}.

\section{Conclusion}
\label{sec:6}
In this paper, we propose an efficient termination mechanism of the MA method to find the error locator polynomials by only $2e$ steps in the error correction decoding process of RS codes, where $e$ is the number of errors.
Based on the new termination mechanism, we propose two types of error correction decoding algorithms for LCH-FFT-based RS codes, which have lower multiplication complexity than the existing correlation decoding algorithms.

\ifCLASSOPTIONcaptionsoff
  \newpage
\fi

\bibliographystyle{IEEEtran}
\bibliography{CNC-v1}

\end{document}